\documentclass{article}
\usepackage{graphicx,amssymb,amsmath}
\usepackage[margin=1in]{geometry}

\usepackage{hyperref,cite}
\usepackage{microtype}

\makeatletter

\let\c@theorem\relax

\makeatother

\usepackage{amsthm}
\usepackage[capitalize,nameinlink]{cleveref}

\newtheorem{theorem}{Theorem}
\newtheorem{lemma}[theorem]{Lemma}
\newtheorem{corollary}[theorem]{Corollary}
\newtheorem{observation}[theorem]{Observation}
\newtheorem{definition}[theorem]{Definition}

\usepackage{lineno}

\title{Angles of Arc-Polygons and Lombardi Drawings of Cacti}

\author{David Eppstein\thanks{Department of Computer Science, University of California, Irvine, {\tt \{eppstein,dfrishbe,mosegued\}@uci.edu}}
        \and Daniel Frishberg$^*$ \href{https://orcid.org/0000-0002-1861-5439}{\includegraphics[height=1em]{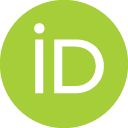}} \and Martha C. Osegueda$^*$ \href{https://orcid.org/0000-0002-1077-1074}{\includegraphics[height=1em]{ORCIDiD_icon128x128}}}

\index{Eppstein, David}
\index{Frishberg, Daniel}
\index{Osegueda, Martha C.}

\begin{document}
\thispagestyle{empty}
\maketitle

\begin{abstract}
We characterize the triples of interior angles that are possible in non-self-crossing triangles with circular-arc sides, and we prove that a given cyclic sequence of angles can be realized by a non-self-crossing polygon with circular-arc sides whenever all angles are $\le\pi$. As a consequence of these results, we prove that every cactus has a planar Lombardi drawing (a drawing with edges depicted as circular arcs, meeting at equal angles at each vertex) for its natural embedding in which every cycle of the cactus is a face of the drawing. However, there exist planar embeddings of cacti that do not have planar Lombardi drawings.
\end{abstract}

\section{Introduction}
Artist Mark Lombardi drew beautiful diagrams of international political and financial conspiracies, often using curved edges and circular layouts~\cite{HobRic-03}. His name is commemorated in Lombardi drawing, a style of graph drawing in which the edges are drawn as circular arcs that meet at equal angles at each vertex~\cite{DunEppGoo-DCG-13,DunEppGoo-JGAA-12}. Many kinds of graph are now known to have Lombardi drawings, including for instance all 2-degenerate graphs, the graphs that can be formed from a single edge by repeatedly adding a new vertex incident to at most two previous vertices~\cite{DunEppGoo-JGAA-12}. Beyond their aesthetic quality, these drawings can be considerably more compact than straight-line drawings of the same graphs~\cite{DunEppGoo-DCG-13}, and they have found application in the realization of soap bubble foams~\cite{Epp-DCG-14} and in the visualization of knots and links~\cite{KinKobLof-JoCG-19}.

In this style of drawing, it is of interest, when possible, to avoid any crossings or intersections of edges other than at shared endpoints, forming a planar drawing~\cite{DunEppGoo-JoCG-18}. Obviously, this requires that the graph to be drawn be planar, but not every planar graph has a planar Lombardi drawing~\cite{DunEppGoo-JGAA-12}. The planar graphs known to have planar Lombardi drawings include the trees~\cite{DunEppGoo-DCG-13}, the subcubic planar graphs~\cite{Epp-DCG-14}, the 4-regular polyhedral graphs~\cite{Epp-DCG-14}, the Halin graphs~\cite{DunEppGoo-JGAA-12}, and the outerpaths~\cite{DunEppGoo-JoCG-18}. However, examples of planar graphs with no planar Lombardi drawing have been found for graph classes including the 4-regular graphs~\cite{Epp-DCG-14}, planar 3-trees~\cite{DunEppGoo-JoCG-18}, and planar bipartite graphs~\cite{Epp-CCCG-19}. In addition, for planar graphs with a fixed choice of embedding, a planar Lombardi drawing might not exist even when the given graph is series-parallel~\cite{Epp-CCCG-19}.

For some of the simplest classes of planar graphs (notably, the outerplanar graphs) the existence of planar Lombardi drawings has remained unknown. In this work, we tackle an even simpler class of graphs, the cacti. Intuitively, a cactus is a tree of cycles; it can be defined as a graph in which each edge belongs to at most one cycle. These graphs have a natural class of planar embeddings in which their cycles form faces of the embedding, with the rest of the graph always drawn outside of the cycle. In this paper, we show that these embeddings of cacti always have Lombardi drawings. However, we find examples of other embeddings of cacti (including one as simple as a triangle with four leaf vertices attached to each triangle vertex) that have no planar Lombardi drawing.

The cycles of any planar Lombardi drawing form simple closed curves in the plane, with sides composed of circular arcs; we call these shapes arc-polygons. The constraints of a Lombardi drawing translate into constraints on the vertex angles of these arc-polygons, and naturally raise the question of which systems of angles can be realized by an arc-polygon and what other geometric properties their realizations have. For instance, the analysis of arc-quadrilaterals with equal angles at all vertices, and the key property of these arc-quadrilaterals that their vertices are all cocircular, figured heavily into our previous work on Lombardi drawings of 4-regular graphs~\cite{Epp-DCG-14}, planar bipartite graphs, and embedded series-parallel graphs~\cite{Epp-CCCG-19}. Here, we focus on arc-triangles, the simplest (and therefore most highly constrained) shapes needed for the faces of Lombardi drawings of cacti. We completely characterize the triples of angles that can be realized by simple arc-triangles; this characterization forms the basis of our proof that some embedded cacti have no planar Lombardi drawing. We also use the same characterization to prove a natural sufficient condition for the existence of a simple arc-polygon with specified interior angles: such an arc-polygon exists whenever all of the specified angles are at most $\pi$. Larger angles than this are not needed for the Lombardi drawings of the natural embeddings of cacti, from which it follows that all cacti have planar Lombardi drawings for their natural embeddings.

\subsection{New results}
The main results that we prove in this work are the following.
\begin{itemize}
\item We completely characterize the triples of angles that can be realized as the internal angles of arc-triangles, in terms of a system of linear inequalities on the angles (\cref{thm:triangle}).
\item We prove that every cyclic sequence of three or more angles, all of which are in the range $[0,\pi]$, except for the triple $(0,0,\pi)$, can be realized as the internal angles of an arc-polygon (\cref{thm:up-to-pi}).
\item We prove that every cactus graph has a planar Lombardi drawing in its natural planar embedding, the embedding in which every cycle is a face (\cref{thm:natural-cactus}).
\item We find an embedded cactus graph that does not have a planar Lombardi drawing for that embedding (\cref{thm:bad-hat}).
\end{itemize}

\subsection{Related work}
As well as in Lombardi drawing, circular arcs have been incorporated in other ways into graph drawing; see for instance~\cite{CheDunGoo-DCG-01,EfrErtKob-JGAA-07,AngEppFra-JGAA-14,Schu-JGAA-15,CarHofKus-CG-18}. Force-directed methods can often achieve good but not perfect angular resolution for arc-based graph drawings~\cite{CheCunGoo-GD-11,BanEppGoo-GD-12,DonFuXu-CVS-13}, and the effectiveness of curved edges in graph drawing has been tested through user studies~\cite{XuRooPas-TVCG-12,PurHamNol-GD-12,Dud-PhD-16}.

The cactus graphs whose drawings we study here are an old and well-studied class of graphs~\cite{Hus-JCP-50,HarUhl-PNAS-53}. Cactus graphs are one of the most basic minor-closed families of graphs, defined (as simple graphs) by the absence of a diamond graph minor or (as multigraphs) by the absence of a 3-edge dipole graph minor. In graph drawing, cactus graphs play a key role in the proof that polyhedral graphs have greedy embeddings~\cite{LeiMoi-DCG-10}, in the visualization of minimum cuts~\cite{BraCorFie-CGTA-04}, and in the approximation of maximum planar subgraphs~\cite{CalFerFin-Algs-98}. The question of whether cactus graphs have Lombardi drawings was explicitly posed by our previous work on the non-existence of Lombardi drawings for certain bipartite planar graphs~\cite{Epp-CCCG-19}; this question was the main motivation for our present work.

\begin{figure}[t]
\centering\includegraphics[width=0.3\columnwidth]{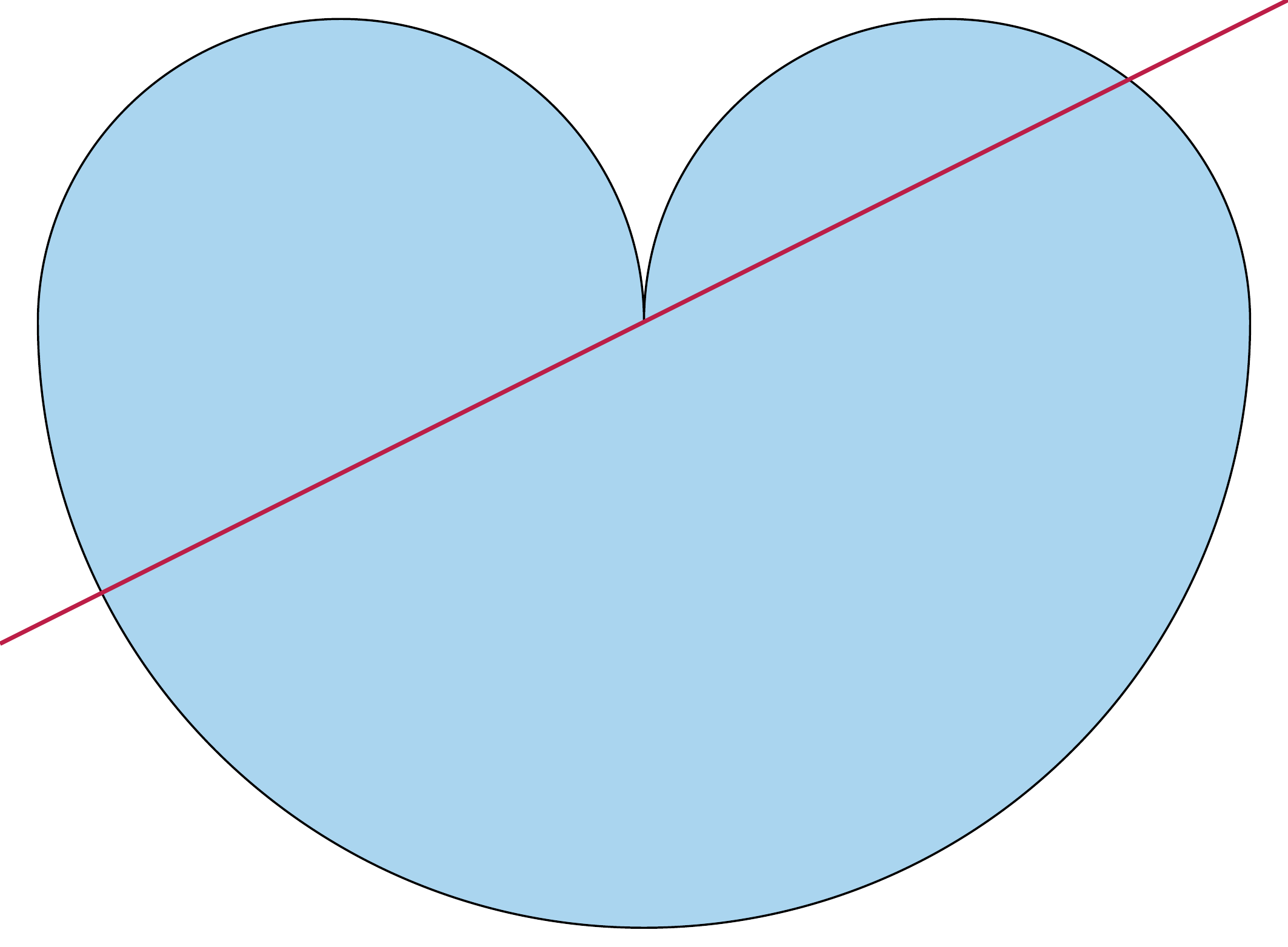}
\caption{Roger Boscovich's arc-triangle formed from three semicircles has the property that every line through its cusp partitions the perimeter into two equal lengths.}
\label{fig:boscovich}
\end{figure}

Beyond our work's contributions to graph drawing, we believe that it contributes to the fundamental study of arc-polygons. These are a natural and important class of two-dimensional shapes whose long history of study can be traced back to the work of Archimedes and Pappus on the arbelos, an arc-triangle formed by three semicircles on the same side of a line, to the use of the Reuleaux triangle in Gothic architecture and by Leonardo da Vinci for map projection and fortress floor plans, and to the work of 18th-century mathematician Roger Boscovich on shapes that have a center through which every line bisects the perimeter (\cref{fig:boscovich})~\cite{BanGib-AMM-94}. It has been stated that ``more than 90\% of machined parts'' have arc-polygon shapes, because of their ease of manufacturing~\cite{CheVenWu-IVC-96}, and arc-polygons have been used to model the shapes for irregular parts to be packed into and cut from metal sheets~\cite{PlaTseShy-ICTME-20}. Arc-polygons can accurately approximate arbitrary smooth curves~\cite{Fej-BAMS-48,MeeWal-JCAM-95}, and their approximations of non-smooth curves form a useful stepping stone to the Riemann mapping theorem~\cite{Kel-BAMS-26}. Aichholzer et al.~\cite{AicAurHac-IJCGA-11} argue that for approximating irregular shapes in this way, arc-polygons have significant advantages over straight-sided polygons in allowing more concise and accurate representations while still allowing efficient computations of basic geometric primitives such as medial axes. For additional work on the computational geometry of these shapes see~\cite{AicAigAur-JGAA-15,WeiJutAur-EuroCG-18,WanLinFan-CAD-17}.

\section{Preliminaries}
\subsection{Arc-polygons}
\begin{definition}
We define an \emph{arc-polygon} to be a cyclic sequence that alternates between points and closed circular arcs in the Euclidean plane, with each arc appearing in the sequence consecutively with its two endpoints. (Here, we allow line segments to count as a degenerate special case of circular arcs.) The points of the arc-polygon are called its \emph{vertices} and the arcs are called its \emph{edges}. In particular, an \emph{arc-triangle} is an arc-polygon with three vertices and three edges. An arc-polygon is \emph{simple} if the only points of intersection between pairs of its edges are shared vertices consecutive with both edges. The union of the arcs in a simple arc-polygon forms a Jordan curve, which separates the plane into two components, the \emph{interior} and \emph{exterior} of the arc-polygon. We define the \emph{interior angle} of a vertex of a simple arc-polygon to be the angle between tangent lines to its two arcs at that vertex, spanning the interior of the arc-polygon in a (possibly empty) neighborhood of the vertex.

Unlike classical straight-sided polygons, it is also possible to form arc-polygons with only two vertices and two sides.
As in \cite{Epp-CCCG-19}, we call these \emph{bigons}.
\end{definition}

A bigon with two different arcs as sides is automatically simple: two different circles can cross in at most two points, and these two crossing points are used up by the vertices of the bigon, so no more crossings are possible. The two interior angles of a bigon must be equal, and can be any angle in the open interval $(0,2\pi)$. It will be convenient to consider an arc-polygon with two identical arcs to be a kind of degenerate bigon, with interior angle $0$.

\subsection{Lombardi drawings}
\begin{definition}
We define a \emph{Lombardi drawing} of a given graph to be a mapping from the vertices of the graph to points in the plane, and from the edges to circular arcs or straight line segments, with the following two properties:
\begin{itemize}
\item For each edge, the two endpoints of the edge in the graph are mapped to the two endpoints of its arc in the plane.
\item For each vertex $v$, the incident edges form arcs that are equally spaced around the point representing $v$, forming angles of $2\pi/\operatorname{deg}(v)$ between each consecutive pair of arcs.
\end{itemize}
A Lombardi drawing is \emph{planar} if the only points of intersection between pairs of its arcs are shared vertices.
\end{definition}

\subsection{M\"obius transformations}
By the \emph{extended plane} we mean the Euclidean plane augmented with a single point at infinity, denoted $\infty$.
An \emph{inversion} of the extended plane, with respect to a circle~$C$, maps each point $p$ to another point $p'$, so that $p$ and $p'$ both belong to a single ray from the center of $C$, with the product of their distances from the center equal to the squared radius of $C$. The center of $C$ is mapped to $\infty$, and vice versa. We consider a reflection of the plane to be a degenerate case of an inversion, and we consider the line of reflection to be a degenerate case of a circle with infinite radius. A \emph{M\"obius transformation} is any functional composition of inversions. These transformations map circles (or degenerate circles) to other circles (or degenerate circles), and preserve the angles between any two curves. For an introduction to these transformations, and the geometry of circles under these transformations, see e.g. \cite{Sch-79}.

Because M\"obius transformations map circular arcs (or straight line segments) to curves of the same type, and preserve angles, they map simple arc-polygons to other simple arc-polygons and Lombardi drawings to other Lombardi drawings.
It is possible for a M\"obius transformation to map the interior of a simple arc-polygon to its exterior and vice versa, but otherwise these transformations leave the angles of arc-polygons unchanged. As previous work on Lombardi drawing has already shown~\cite{Epp-DCG-14}, these properties make M\"obius transformations very convenient as a tool for bringing pieces of arc-polygons and of Lombardi drawings into a position where they can be glued together.

\section{Arc-triangles}

\begin{figure}[t]
\centering\includegraphics[scale=0.4]{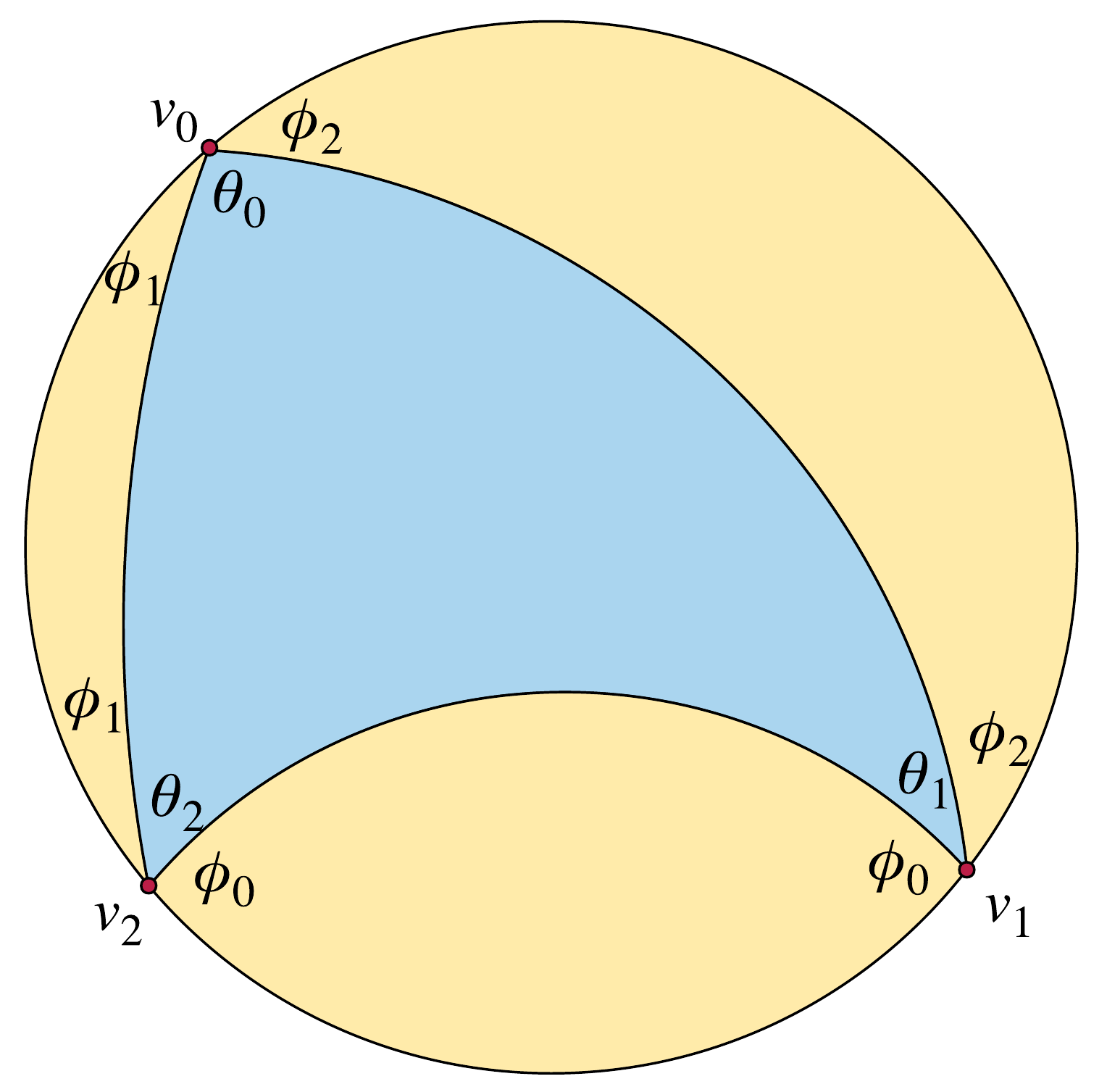}
\caption{A simple arc-triangle (blue), the circle through its three vertices (yellow), and the three bigons between the arc-triangle and the circle, labeled with vertices $v_i$, arc-triangle angles $\theta_i$, and bigon angles $\phi_i$.}
\label{fig:starfleet}
\end{figure}

Consider an arbitrary simple arc-triangle with vertices $v_i$ and interior angles $\theta_i$ (for $i\in\{0,1,2\}$), for instance the one in \cref{fig:boscovich} (for which the interior angles are $2\pi,\pi,\pi$) or the one in \cref{fig:starfleet}. Any three points are contained either in a unique circle or a straight line, but for the following definitions we need to know which side of the circle is its inside and which its outside, so we assume (by perturbing the triangle by a M\"obius transformation, if necessary) that the three vertices are not collinear, and are clockwise as $v_0$, $v_1$, and $v_2$ on the circle $C$ through them. We assume also that these three vertices have the same clockwise ordering on the arc-triangle, as shown in \cref{fig:starfleet}, meaning that when traveling along the arcs from $v_0$ to $v_1$, from $v_1$ to $v_2$, and from $v_2$ to $v_0$, the polygon is consistently on the right side of each arc; if necessary, this can be achieved by an inversion with respect to $C$.

Let $\theta_i$ be the interior angle of the arc-triangle at vertex $v_i$, as labeled on the figure. Each arc of the triangle is separated from $C$ by a bigon, and (in the case that the arc is contained in $C$) we let $\phi_i$ denote the angle of the bigon opposite vertex $v_i$. It is also possible for an arc to lie on $C$, defining a degenerate bigon with $\phi_i=0$. If an arc of the arc-triangle lies outside $C$, it still defines a bigon, but in this case we define $\phi_i$ to be a negative number, the negation of the interior angle of the bigon.

\begin{observation}
\label{obs:linear}
For angles $\theta_i$ and $\phi_i$ defined as above from a simple arc-triangle, and for $i\in\{0,1,2\}$,
\[ \theta_i + \phi_{(i-1)\bmod{3}} + \phi_{(i+1)\bmod{3}} = \pi. \]
\end{observation}

\begin{proof}
The three angles on the left hand side are the angles at $v_i$ measured clockwise from the arc of $C$ clockwise of $v_i$ to the side of the arc-triangle clockwise of $v_i$, from this side of the arc-triangle to the other side, and from the other side of the arc-triangle to the arc of $C$ counter-counterclockwise of $v_i$. Therefore, their sum is the total angle at $v_i$ between the two arcs of $C$, which is just $\pi$.
\end{proof}

\begin{corollary}
\label{cor:phi}
Let $\psi=(\pi-\sum\theta_i)/2$. Then $ \phi_i = \psi+\theta_i$.
\end{corollary}

\begin{proof}
Simple algebra verifies that this is the solution to the system of three linear equations in three unknowns given by \cref{obs:linear}.
\end{proof}

\begin{lemma}
\label{lem:1bad}
For three given angles $\theta_i$ with $0\le\theta_i\le 2\pi$, at most one of the angles $\phi_i$ calculated from the formula of \cref{cor:phi} can fail to satisfy $-\pi\le\phi_i\le\pi$.
If $\phi_i<-\pi$ then $\theta_i$ must be the only angle of the three given angles that is less than $\pi$; if $\phi_i>\pi$ then $\theta_i$ must be the only angle of the three given angles that is greater than~$\pi$.
\end{lemma}

\begin{proof}
Consider the three angles $\theta_i$ and the angle $\pi$, listed in sorted order. Each $\phi_i$ is obtained by adding two of these angles, subtracting the other two, and dividing the total by two; in terms of the sorted ordering of the angles, we can represent this
(up to the sign of $\phi_i$) as one of the three sign patterns $++--$, $+-+-$, or $+--+$.
However, only the first of these can produce a value of $\phi_i$ that is out of range.
The other two sign patterns describe the difference of the smallest two angles, plus or minus the difference of the largest two angles. Because these two differences span different ranges of angles, they can sum to at most $2\pi$ (and must sum to at least $-2\pi$). Therefore, for these other two sign patterns, division by two produces a value in the range between $-\pi$ and~$\pi$.

For the sign pattern $++--$, $\pi$ cannot be the maximum or minimum of the three angles, because if it were then each difference of angles would be at most $\pi$, and so would the average of two differences. So when this sign pattern leads to a value of $\phi_i$ that is out of range, there is a unique angle $\theta_i$ with the same sign in the sign pattern, implying the second part of the lemma. 
\end{proof}

\begin{figure}[t]
\includegraphics[width=0.5\columnwidth]{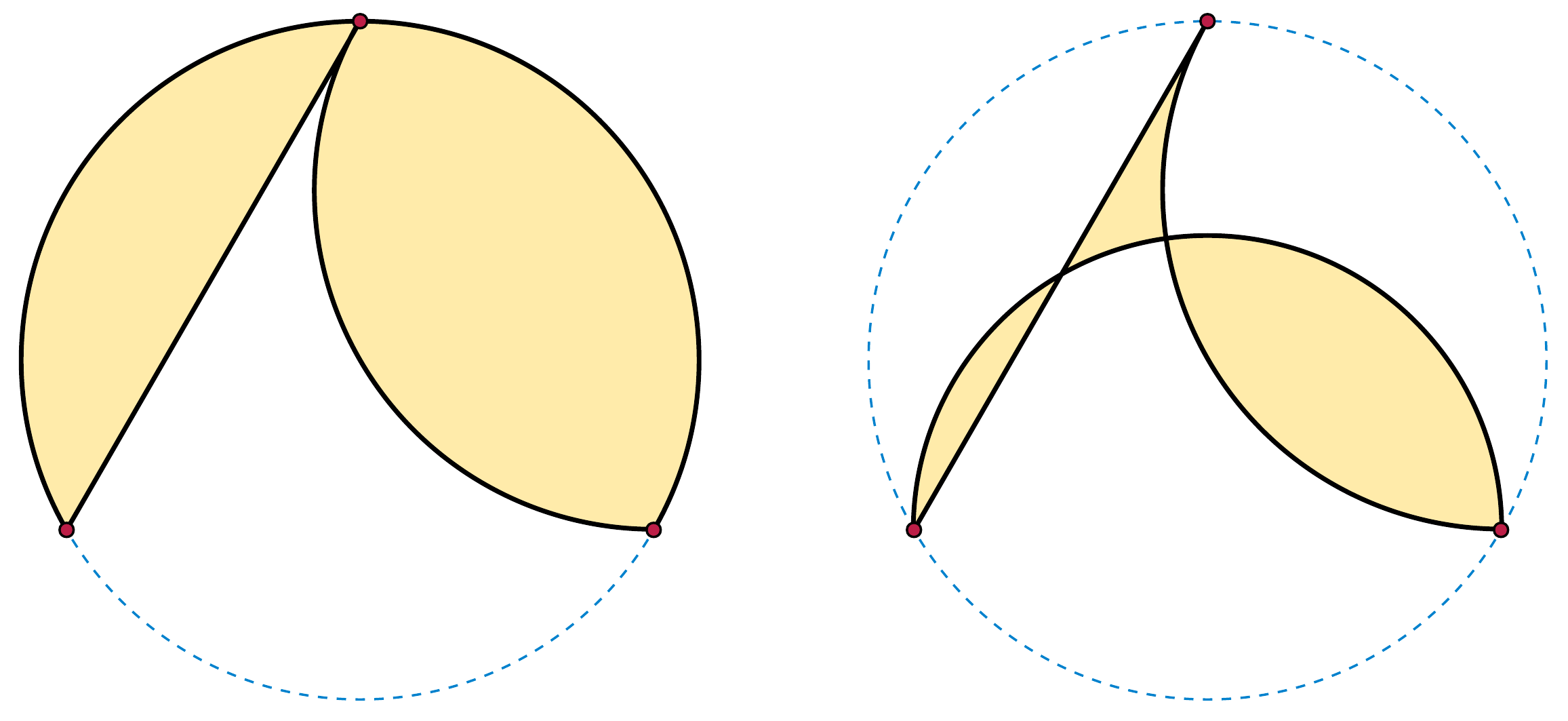}
\caption{For some triples of angles $\theta_i$, using \cref{cor:phi} to compute angles $\phi_i$ and then drawing arcs with these angles may not produce a simple arc-triangle. Left: $\theta_i=0,4\pi/3,5\pi/3$ (at the top, right, and left vertices, respectively) and $\phi_i=-\pi,\pi/3,2\pi/3$. The arc with angle $\phi_i=-\pi$ overlays the opposite vertex. Right: $\theta_i=0,3\pi/2,11\pi/6$, $\phi_i=-7\pi/6,\pi/3,2\pi/3$. The arc with angle $\phi_i=-7\pi/6$ crosses the other two arcs.}
\label{fig:crossed}
\end{figure}

The definitions above of $\theta_i$ and $\phi_i$ are only for simple arc-triangles, but one can also plug in other choices of $\theta_i$, compute $\phi_i$ from them, and examine the arc-triangles that result, which may not be simple. Examples of what can go wrong are depicted in \cref{fig:crossed}: the left arc-triangle of the figure has a vertex on the opposite side arc, and the right arc-triangle has two crossing pairs of arcs.

\begin{theorem}
\label{thm:triangle}
Three given angles $\theta_i$ with $0\le\theta_i\le 2\pi$ can be realized as the interior angles of a simple arc-triangle if and only if the angles $\phi_i$, as calculated by \cref{cor:phi}, satisfy the inequalities $-\pi<\phi_i<\pi$.
\end{theorem}

\begin{proof}
By \cref{cor:phi}, a realization must have the form shown in \cref{fig:starfleet}: three circular arcs making angles of $\phi_i$ to the circumcircle of the three vertices. The location of these vertices on the circumcircle, and the location of the circumcircle, are not important, as any three points may be transformed into any other three points by a M\"obius transformation. The important question for the realizability of these angles by an arc-triangle is whether these three arcs form a simple arc-triangle $\Delta$, or whether they have undesired crossings.

First, let us disprove the existence of a realization when the angles $\phi_i$ do not satisfy the inequalities. 

Consider the case that some $\phi_i=\pm\pi$. In this case, the arc with angle $\phi_i$, opposite vertex $v_i$, coincides with an arc of the circumcircle passing through vertex $v_i$. Since a simple arc-triangle is not allowed to have one of its sides passing through the opposite vertex, the angles in this case cannot come from a simple arc-triangle.

Next, instead suppose that there is an angle $\phi_i$ that is out of range: $\phi_i<-\pi$ or $\phi_i>\pi$.
By \cref{lem:1bad}, it must be the only angle with this property. In the case that $\phi_i<-\pi$, consider continuously increasing $\phi_i$ (and decreasing the two given angles that are not $\theta_i$ to match) until it is greater than $-\pi$.
By the second part of \cref{lem:1bad}, this decrease in the given angles cannot cause them to become negative. At the end of the sequence, the angles meet the conditions of the previous case of the theorem, producing a simple arc-triangle, but at the point in the sequence where $\phi_i$ becomes equal to $-\pi$, and the arc with angle $\phi_i$ crosses over vertex $v_i$, two crossings with the other two arcs are removed.  There is no combinatorial change to the configuration of arcs and their crossings anywhere else in the sequence, so these two crossings must have been present in the configuration with the given angle~$\phi_i$. In the case when $\phi_i>\pi$, instead continuously decrease $\phi_i$ until it is less than $\pi$; the remaining argument is the same.

Finally, let us prove that angles $\phi_i$ satisfying the inequalities must be simple. Suppose all $\phi_i$ obey $-\pi<\phi_i<\pi$. Then it is impossible for the two arcs at $v_i$ to cross.
If the values of $\phi$ corresponding to these two arcs have opposite signs, then one is inside the circumcircle and the other outside, and if one has sign zero, then it lies on an arc of the circumcircle inaccessible to the other. If both arcs have values of $\phi$ with the same sign, then they leave the circumcircle at $v_i$ in a relative ordering, determined by their angle at $v_i$, that is consistent with the relative ordering at which they reach the circumcircle again at the other two vertices, determined by the positions of those two vertices. Since circular arcs can meet only twice, and these two arcs meet once at $v_i$, any additional crossing would swap their relative positions, so they cannot cross.
\end{proof}

\begin{figure}[t]
\centering\includegraphics[width=0.3\columnwidth]{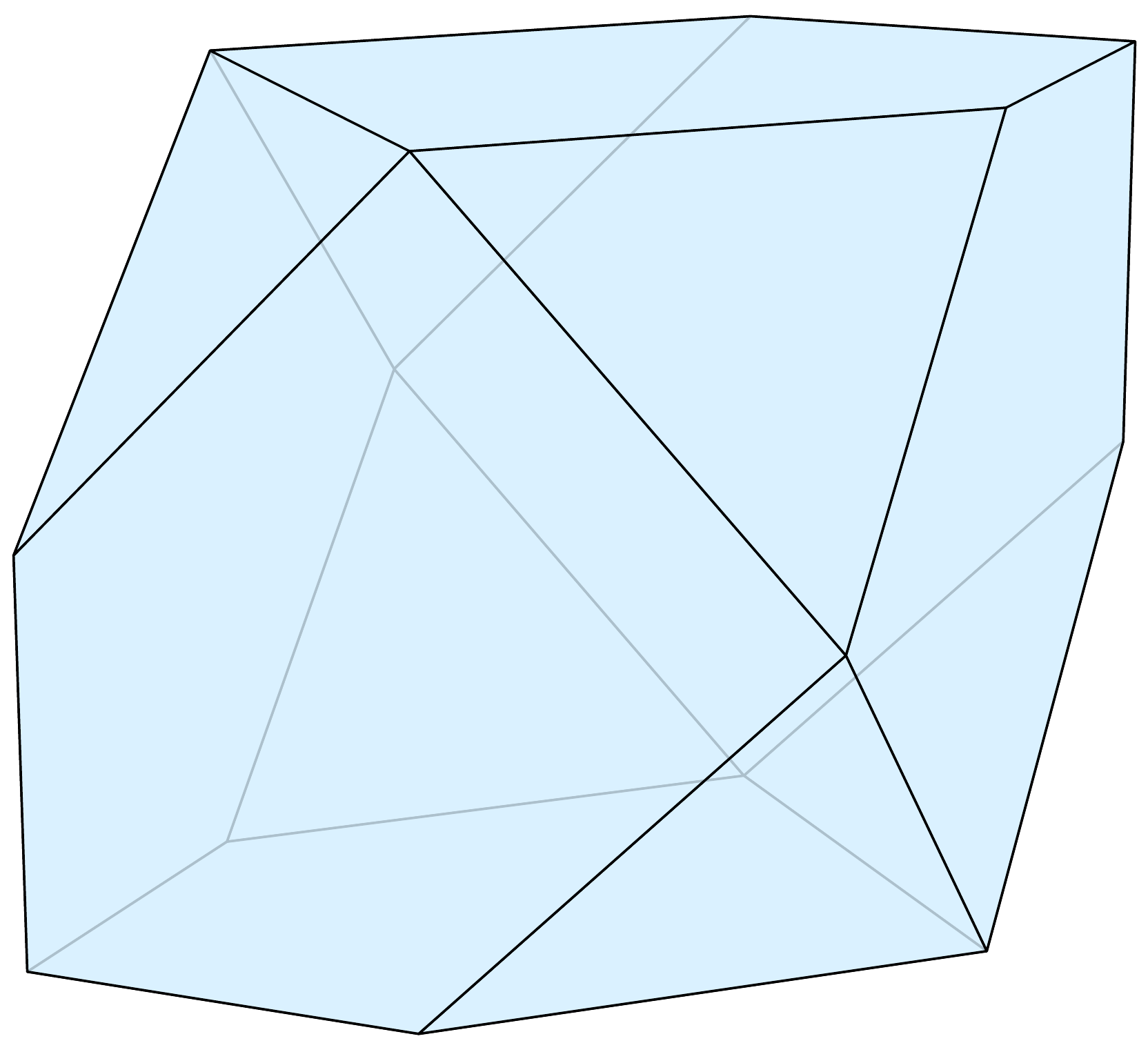}
\caption{Polyhedral visualization of the subset of triples of angles that can be realized by simple arc-triangles.}
\label{fig:feasible}
\end{figure}

The linear inequalities $0\le\theta_i\le 2\pi$ can be thought of as defining a cube in a three-dimensional space having the angles $\theta_i$ as Cartesian coordinates. For this interpretation of these numbers, the additional linear inequalities $-\pi<\phi_i<\pi$ cut off six of the eight corners of the cube (the corners where not all coordinate values are equal), replacing them by equilateral triangle faces whose vertices are midpoints of the cube edges. The result of this truncation is a feasible region of triples of angles that can be realized by simple arc-triangles, bounded by six pentagonal faces (the truncated faces of the cube) and six triangular faces (the inequalities on $\phi_i$). It is depicted in \cref{fig:feasible}. The result in \cref{lem:1bad} that only one of the inequalities on $\phi_i$ can be exceeded corresponds geometrically to the fact that the triangular faces meet only in vertices, not in edges. At the vertices where two triangular faces meet, two of these inequalities on $\phi_i$ are exactly met but neither is exceeded.

We remark that it is straightforward to implement this construction of a simple arc-triangle for given angles and given triangle vertices, by computing the circumcircle of the vertices and the angle made by the triangle arcs to this circumcircle.

\section{Arc-polygons}

We do not have a complete characterization of the angle sequences of higher-order simple arc-polygons, as we do for arc-triangles. Instead, we prove a sufficient condition, which will be enough for our application to Lombardi drawing: with a single exception, an angle sequence can be realized by a simple arc-polygon whenever all angles are $\le\pi$. For triangles, this follows from our previous characterization:

\begin{lemma}
\label{lem:polygon-base}
If a triple of given angles $\theta_i$ all lie in the range $0\le\theta_i\le\pi$ and is not a permutation of $(0,0,\pi)$ then there exists a simple arc-triangle having these interior angles.
\end{lemma}

\begin{proof}
For these angles, the values of $\phi_i$ calculated according to \cref{cor:phi} range from a minimum of $-\pi/2$ (when $\theta_i=0$ and both other given angles are $\pi$) to a maximum of $\pi$ (when $\theta_i=\pi$ and both other given angles are $0$). Because the angles $(0,0,\pi)$ are the only combination achieving this maximum, all other triples of angles result in $-\pi/2\le\phi_i<\pi$. The result follows from \cref{thm:triangle}.
\end{proof}

We also need the following two special cases:

\begin{figure}[t]
\centering\includegraphics[width=0.5\columnwidth]{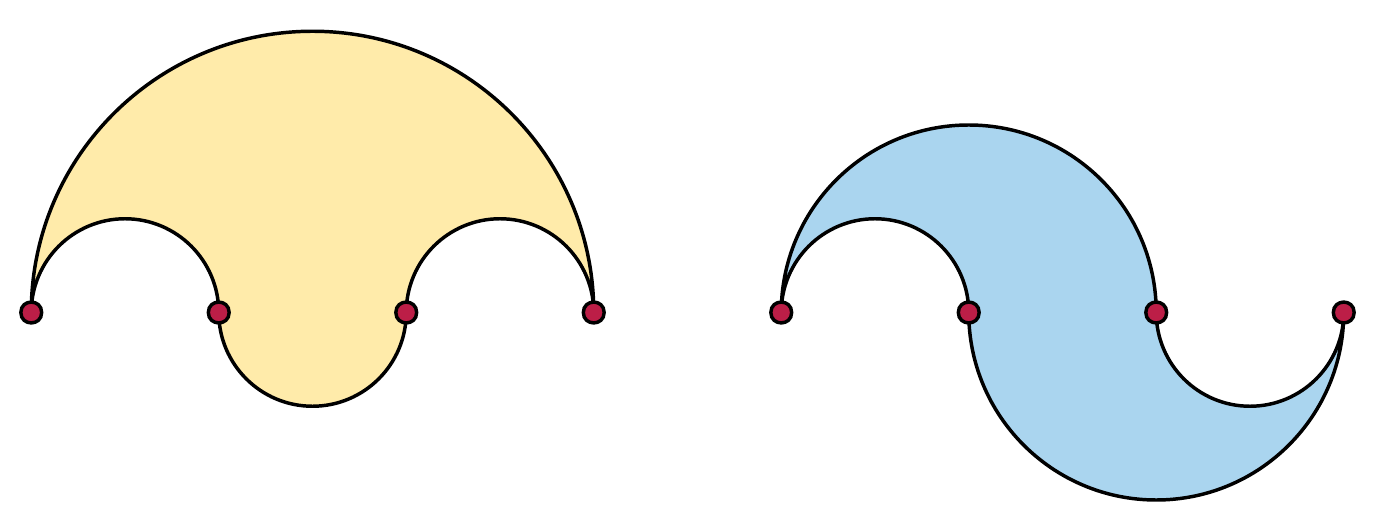}
\caption{Arc-quadrilaterals with interior angles $(0,0,\pi,\pi)$ (left, yellow) and $(0,\pi,0,\pi)$ (right, blue).}
\label{fig:kilroy}
\end{figure}

\begin{observation}
\label{obs:kilroy}
The sequences $(0,0,\pi,\pi)$ and $(0,\pi,0,\pi)$ can be realized as the interior angles of arc-quadrilaterals.
\end{observation}

\begin{proof}
See \cref{fig:kilroy}, which realizes these sequences using four vertices equally spaced along a line, with semicircular arcs.
\end{proof}

For higher-order polygons, we will prove the existence of a simple realization by induction on the order of the polygon, gluing together arc-polygons with fewer vertices at \emph{cusps}, vertices of interior angle zero.

\begin{definition}
If $\sigma$ and $\tau$ denote sequences of angles, we let $\sigma\tau$ denote their concatenation, and we let $\sigma0$ and $\tau0$ denote the sequences of angles obtained from $\sigma$ and $\tau$ respectively by including one more angle equal to zero.
\end{definition}

With this notation, the following lemma describes the process of gluing together two arc-polygons.

\begin{lemma}
\label{lem:polygon-glue}
Let $\sigma$ and $\tau$ both be sequences of angles such that the augmented sequences $\sigma0$ and $\tau0$ are realizable as the sequences of interior angles of simple arc-polygons. Then the concatenation $\sigma\tau$ is also realizable by a simple arc-polygon.
\end{lemma}

\begin{proof}
Let $S$ and $T$ be simple arc-polygons realizing $\sigma_0$ and $\tau_0$ respectively. Apply inversions separately to~$S$ and $T$, centered respectively on the cusps of $S$ and $T$, producing arc-polygons $S'$ and $T'$. Both $S'$ and $T'$ have their cusp transformed to lie at $\infty$, and the circular arcs incident to these cusps transformed into infinite rays. All other vertices and edges of $S'$ and $T'$ remain finite. The condition that the angle at the cusp be zero in $S$ and $T$ corresponds, in $S'$ and $T'$, to the condition that these two rays be parallel. By adjusting the radii of the inversions we can additionally ensure that in both $S'$ and $T'$ the rays lie on lines at distance one from each other. By rotating $S'$ and $T'$ so that their pairs of rays both lie on the same two lines, in opposite directions, and translating them far enough apart along these two parallel lines, we may glue them together into a single simple arc-polygon, as shown in \cref{fig:glue-infinite-cusps}, giving a realization of $\sigma\tau$.
\end{proof}

\begin{figure}[t]
\centering\includegraphics[width=0.4\columnwidth]{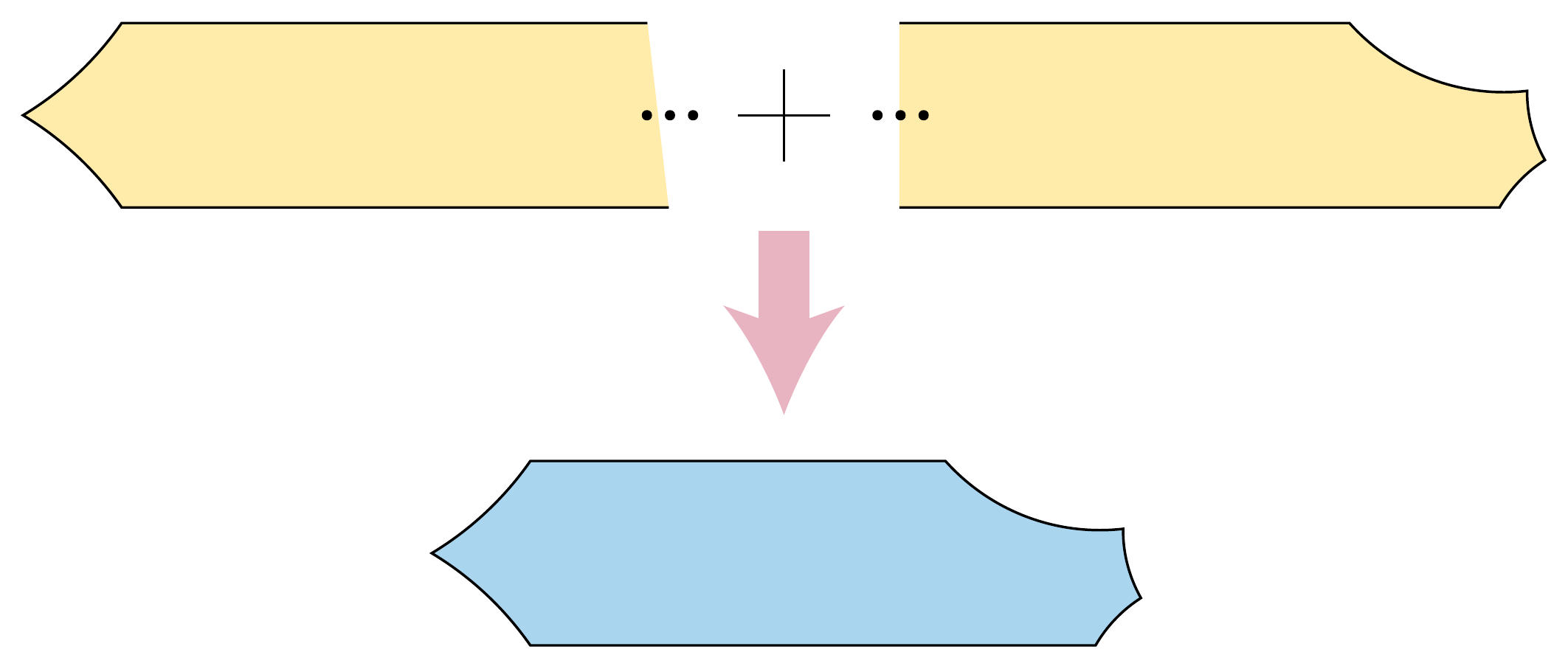}
\caption{Gluing together two simple arc-polygons with cusps at infinity, realizing angle sequences $\sigma0$ and $\tau0$, to produce a single arc-polygon realizing angle sequence~$\sigma\tau$.}
\label{fig:glue-infinite-cusps}
\end{figure}

\begin{theorem}
\label{thm:up-to-pi}
A finite sequence of three or more angles $\theta_i$ can be realized as the interior angles of a simple arc-polygon whenever all angles satisfy $0\le\theta_i\le\pi$ and the sequence is not a permutation of $(0,0,\pi)$.
\end{theorem}

\begin{proof}
We use induction on the number of angles, with \cref{lem:polygon-base} as a base case for sequences of exactly three angles and \cref{obs:kilroy} for the two four-angle-sequences $(0,0,\pi,\pi)$ and $(0,\pi,0,\pi)$ and their cyclic permutations.

For any other sequence of more than three angles that includes $\pi$ as one of its angles, we form a shorter sequence by removing the $\pi$, realize the shorter sequence as the interior angles of a simple arc-polygon by the induction hypothesis, and then reinsert the vertex with angle $\pi$ anywhere along the circular arc between its two neighbors in the resulting arc-polygon.

For a sequence of more than three angles that does not include $\pi$ as an angle, partition the sequence into a concatenation $\sigma\tau$ where both $\sigma$ and $\tau$ contain at least two angles. By the induction hypothesis, both $\sigma0$ and $\tau0$ are realizable, so by \cref{lem:polygon-glue} their concatenation $\sigma\tau$ is also realizable.
\end{proof}

It is tempting to guess that all sequences of four or more angles are realizable by simple arc-polygons, but this is not true, as the following example demonstrates.

\begin{figure}[t]
\centering\includegraphics[width=0.3\columnwidth]{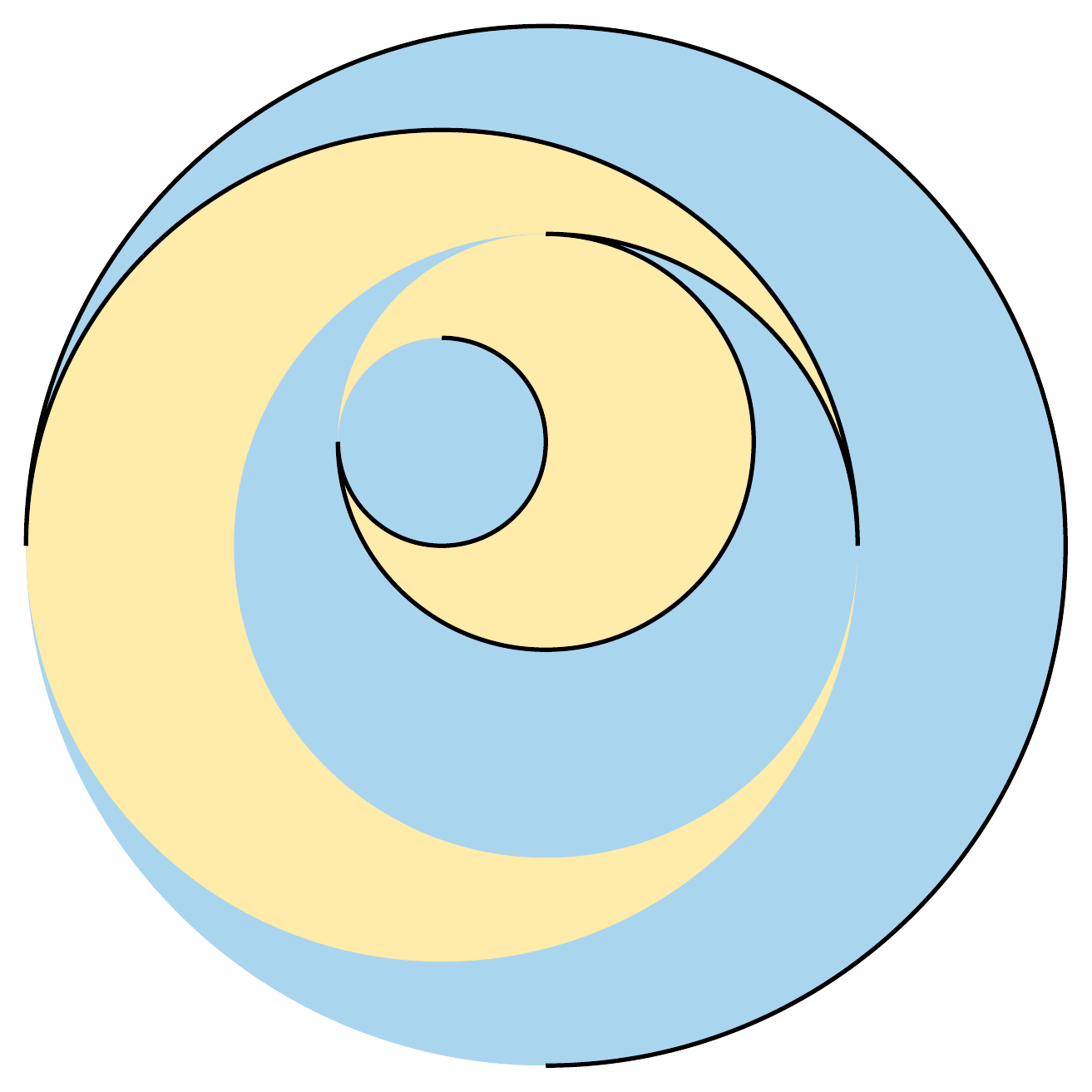}
\caption{Interior angles that alternate between $0$ and $2\pi$ force the arcs of a simple arc-polygon to be confined to a sequence of nested circles.}
\label{fig:moons}
\end{figure}

\begin{observation}
In a simple arc-polygon, if the two interior angles at the vertices of a given arc $A$ are $0$ and $2\pi$, then the two adjacent arcs lie on opposite sides of the circle through $A$.
\end{observation}

\begin{proof}
The arcs are tangent to $A$, and two circles that are tangent cannot cross. The choice of angles ensures that, near the point of tangency, the arcs are on opposite sides of the circle, and because they cannot cross this circle they remain on opposite sides for their entire length.
\end{proof}

\begin{corollary}
\label{cor:no-zigzag}
It is not possible for a simple arc-polygon with an even number of sides to have angles that alternate between $0$ and $2\pi$ around the polygon.
\end{corollary}

\begin{proof}
The alternation would cause the sides to belong to a sequence of nested circles (\cref{fig:moons}), from which it would be impossible to loop back around from the innermost to the outermost circle.
\end{proof}

\section{Lombardi drawing}
\subsection{Existence}
As stated in the introduction, a \emph{cactus} is a graph in which each edge belongs to at most one cycle. In the natural embedding of a cactus, each cycle forms a face, with the other edges that are incident to each cycle vertex placed outside the face. At a vertex of degree $d\ge 2$, the interior angle of the cycle will be $2\pi/d$. Because all angles are nonzero and at most $\pi$, the sequence of interior angles meets the conditions of \cref{thm:up-to-pi}, allowing the face to be drawn as a simple arc-polygon. Each edge belongs to at most one cycle, and the edges that do not belong to cycles are even easier to draw (on their own) as line segments. To construct a Lombardi drawing for the entire graph, we glue these pieces of drawings together using M\"obius transformations, similar to the gluing process used to construct arc-polygons in \cref{lem:polygon-glue}.

\begin{definition}
If $G$ is an embedded graph, and $H$ is a subgraph of $G$ we define a \emph{partial Lombardi drawing} of $H$ with respect to $G$ to be a drawing of the vertices and edges of $H$ as points and circular arcs, in the manner of a Lombardi drawing, with pairs of circular arcs that meet at a vertex having the angle that they would be required to have in a Lombardi drawing of $G$. We define a partial Lombardi drawing to be \emph{planar} in the same way as for a full Lombardi drawing: the only intersection points of arcs are shared endpoints.
\end{definition}

The \emph{biconnected components} of a cactus are its cycles and its edges that do not belong to cycles. Its \emph{articulation points} are the vertices that belong to more than one biconnected component. The reasoning above proves the following observation:

\begin{observation}
\label{obs:cactus-blocks}
Each biconnected component of a cactus has a planar partial Lombardi drawing with respect to the natural embedding of the cactus.
\end{observation}

\begin{lemma}
\label{lem:cactus-glue}
Let $H_i$ be a collection of subgraphs of a given graph $G$ that all share a common vertex $v$, with no other pairwise intersections, let $G$ be given a planar embedding in which the edges of each subgraph $H_i$ appear consecutively at $v$, and let each $H_i$ have a planar partial Lombardi drawing with respect to $G$. Then the union $\cup H_i$ also has a planar partial Lombardi drawing with respect to $G$.
\end{lemma}

\begin{proof}
Perform an inversion centered at $v$ on each drawing of $H_i$, producing a drawing where $v$ is at $\infty$ and its incident edges have become rays, radiating outward from the finite part of the drawing, with relative angles equal to the angles they should have at $v$ in a Lombardi drawing of $G$.

Let $v$ have degree $d$ in $G$, and draw a system of $d$ rays meeting at the origin in the plane, separated from each other by angles of $2\pi/d$. Assign consecutive subsets of rays to each subgraph $H_i$, according to its number of edges incident to $v$. Additionally assign to each subgraph $H_i$ a wedge of the plane, with the apex as the origin, extending beyond these assigned rays to an additional angle of $\pi/d$ on both sides. This assignment produces open wedges for each $H_i$, within which (after a suitable rotation) each may be placed; we place each $H_i$ so that its rays are parallel to the rays it is assigned, but we do not require these rays to coincide. Because the wedges are open, they are disjoint from each other and from the origin.

Once all of the drawings of subgraphs have been placed in this way, another inversion centered at the origin returns the drawing to a state in which all vertices and arcs are finite. All of the angles within each drawing of each subgraph $H_i$ have been preserved, and the rays to $\infty$ invert to circular arcs meeting at $v$ at the desired angles.
\end{proof}

\begin{theorem}
\label{thm:natural-cactus}
Every cactus has a planar Lombardi drawing for its natural embedding.
\end{theorem}

\begin{proof}
We construct planar Lombardi drawings for each biconnected component of the cactus by \cref{obs:cactus-blocks}, and glue these partial drawings into partial drawings for larger subgraphs by applying \cref{lem:cactus-glue} at each articulation point of the cactus, until the entire drawing has been glued together. Each step preserves planarity, so the resulting Lombardi drawing is planar.
\end{proof}

\subsection{Nonexistence}
\begin{figure}
\centering\includegraphics[width=0.4\columnwidth]{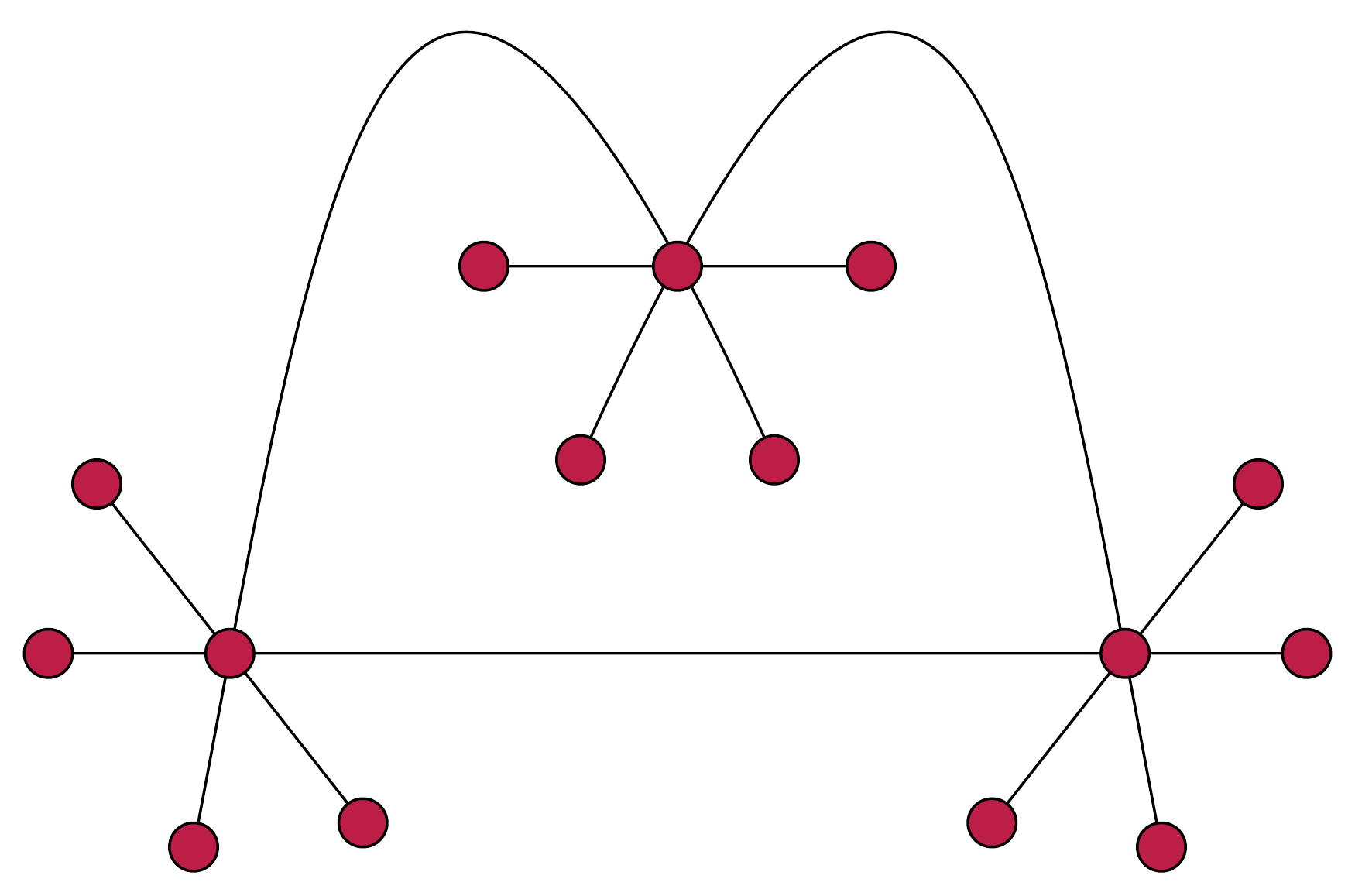}
\caption{An embedded cactus that has no planar Lombardi drawing.}
\label{fig:badhat}
\end{figure}

\begin{theorem}
\label{thm:bad-hat}
The embedded cactus depicted in \cref{fig:badhat}, consisting of a 3-cycle with four pendant vertices at each 3-cycle vertex, embedded so that the pendant vertices are outside the 3-cycle at two 3-cycle vertices and inside at one, has no planar Lombardi drawing.
\end{theorem}

\begin{proof}
A planar Lombardi drawing of this embedded graph would draw its 3-cycle as a simple arc-triangle with angles $\theta_i$ of $\pi/3$, $\pi/3$, and $5\pi/3$. A calculation following \cref{cor:phi} shows that these angles produce corresponding angles $\phi_i$ of $-\pi/3$, $-\pi/3$, and $\pi$, respectively. Since one of the angles $\phi_i$ is $\pi$, \cref{thm:triangle} shows that this simple arc-triangle cannot exist.
\end{proof}

Additional pendant vertices only make the angles farther from realizability, producing an infinite family of examples without planar Lombardi drawings.

\section{Conclusions}

We have completely characterized the triples of internal angles of simple arc-triangles, and proven that for higher-order simple arc-polygons (with one exception) all sequences of angles that are at most $\pi$ are realizable. We have used these results to find planar Lombardi drawings for the natural planar embeddings of all cacti, and to prove that certain other embeddings of cacti do not have planar Lombardi drawings. In particular, this shows that it is possible for an outerplanar graph, embedded in a non-outerplanar way, to fail to have a planar Lombardi drawing.

This work leaves the following questions open for future research:
\begin{itemize}
\item Can we characterize more completely the sequences of angles that can be realized by simple arc-polygons? (\cref{cor:no-zigzag} shows that this question has a nontrivial answer.) What is the computational complexity of finding such a realization?
\item Similarly, what is the computational complexity of determining whether an embedded cactus has a planar Lombardi drawing?
\item In previous work, for trees with fixed embeddings, we were able to prove the existence of Lombardi drawings whose area is a polynomial multiple of the minimum separation between vertices, compared to the exponential area requirement for straight-line drawings with uniformly spaced edges at each vertex~\cite{DunEppGoo-DCG-13}. What are the area requirements for Lombardi drawing of cacti?
\item As already discussed in our previous work on Lombardi drawing, do all outerplanar graphs have planar Lombardi drawings? Do they have planar Lombardi drawings for each outerplanar embedding? What is the computational complexity of finding these drawings?
\end{itemize}


\small
\bibliographystyle{abuser}
\bibliography{lombardi}

\begin{thebibliography}{10}
\urlstyle{rm}

\bibitem{AicAigAur-JGAA-15}
O.~Aichholzer, W.~Aigner, F.~Aurenhammer, K.~{\v{C}}ech~Dobi{\'a}sov{\'a},
  B.~J{\"u}ttler, and G.~Rote.
\newblock {Triangulations with circular arcs}.
\newblock {\em J. Graph Algorithms Appl.} 19(1):43{--}65, 2015,
  \href{http://dx.doi.org/10.7155/jgaa.00346}%
{doi:\nolinkurl{10.7155/jgaa.00346}},
  \href{https://www.ams.org/mathscinet-getitem?mr=3321736}%
{MR3321736}.

\bibitem{AicAurHac-IJCGA-11}
O.~Aichholzer, F.~Aurenhammer, T.~Hackl, B.~J{\"u}ttler, M.~Rabl, and
  Z.~{\v{S}}i{\'r}.
\newblock {Computational and structural advantages of circular boundary
  representation}.
\newblock {\em Internat. J. Comput. Geom. Appl.} 21(1):47{--}69, 2011,
  \href{http://dx.doi.org/10.1142/S0218195911003548}%
{doi:\nolinkurl{10.1142/S0218195911003548}},
  \href{https://www.ams.org/mathscinet-getitem?mr=2777029}%
{MR2777029}.

\bibitem{AngEppFra-JGAA-14}
P.~Angelini, D.~Eppstein, F.~Frati, M.~Kaufmann, S.~Lazard, T.~Mchedlidze,
  M.~Teillaud, and A.~Wolff.
\newblock {Universal point sets for drawing planar graphs with circular arcs}.
\newblock {\em J. Graph Algorithms Appl.} 18(3):313{--}324, 2014,
  \href{http://dx.doi.org/10.7155/jgaa.00324}%
{doi:\nolinkurl{10.7155/jgaa.00324}},
  \href{https://www.ams.org/mathscinet-getitem?mr=3213192}%
{MR3213192}.

\bibitem{BanGib-AMM-94}
T.~Banchoff and P.~Giblin.
\newblock {On the geometry of piecewise circular curves}.
\newblock {\em Amer. Math. Monthly} 101(5):403{--}416, 1994,
  \href{http://dx.doi.org/10.2307/2974900}%
{doi:\nolinkurl{10.2307/2974900}},
  \href{https://www.ams.org/mathscinet-getitem?mr=1272938}%
{MR1272938}.

\bibitem{BanEppGoo-GD-12}
M.~J. Bannister, D.~Eppstein, M.~T. Goodrich, and L.~Trott.
\newblock {Force-directed graph drawing using social gravity and scaling}.
\newblock {\em Proc. 20th Internat. Symp. Graph Drawing (GD 2012)},
  pp.~414{--}425. Springer, Lect. Notes Comput. Sci. 7704, 2012,
  \href{http://dx.doi.org/10.1007/978-3-642-36763-2_37}%
{doi:\nolinkurl{10.1007/978-3-642-36763-2_37}}.

\bibitem{BraCorFie-CGTA-04}
U.~Brandes, S.~Cornelsen, C.~Fie{\ss}, and D.~Wagner.
\newblock {How to draw the minimum cuts of a planar graph}.
\newblock {\em Comput. Geom.} 29(2):117{--}133, 2004,
  \href{http://dx.doi.org/10.1016/j.comgeo.2004.01.008}%
{doi:\nolinkurl{10.1016/j.comgeo.2004.01.008}},
  \href{https://www.ams.org/mathscinet-getitem?mr=2082210}%
{MR2082210}.

\bibitem{CalFerFin-Algs-98}
G.~C{\u{a}}linescu, C.~G. Fernandes, U.~Finkler, and H.~Karloff.
\newblock {A better approximation algorithm for finding planar subgraphs}.
\newblock {\em J. Algorithms} 27(2):269{--}302, 1998,
  \href{http://dx.doi.org/10.1006/jagm.1997.0920}%
{doi:\nolinkurl{10.1006/jagm.1997.0920}},
  \href{https://www.ams.org/mathscinet-getitem?mr=1622397}%
{MR1622397}.

\bibitem{CarHofKus-CG-18}
J.~Cardinal, M.~Hoffmann, V.~Kusters, C.~D. T{\'o}th, and M.~Wettstein.
\newblock {Arc diagrams, flip distances, and Hamiltonian triangulations}.
\newblock {\em Comput. Geom.} 68:206{--}225, 2018,
  \href{http://dx.doi.org/10.1016/j.comgeo.2017.06.001}%
{doi:\nolinkurl{10.1016/j.comgeo.2017.06.001}},
  \href{https://www.ams.org/mathscinet-getitem?mr=3715053}%
{MR3715053}.

\bibitem{CheVenWu-IVC-96}
J.-M. Chen, J.~A. Ventura, and C.-H. Wu.
\newblock {Segmentation of planar curves into circular arcs and line segments}.
\newblock {\em Image Vis. Comput.} 14(1):71{--}83, 1996,
  \href{http://dx.doi.org/10.1016/0262-8856(95)01042-4}%
{doi:\nolinkurl{10.1016/0262-8856(95)01042-4}}.

\bibitem{CheDunGoo-DCG-01}
C.~C. Cheng, C.~A. Duncan, M.~T. Goodrich, and S.~G. Kobourov.
\newblock {Drawing planar graphs with circular arcs}.
\newblock {\em Discrete Comput. Geom.} 25(3):405{--}418, 2001,
  \href{http://dx.doi.org/10.1007/s004540010080}%
{doi:\nolinkurl{10.1007/s004540010080}},
  \href{https://www.ams.org/mathscinet-getitem?mr=1815440}%
{MR1815440}.

\bibitem{CheCunGoo-GD-11}
R.~Chernobelskiy, K.~I. Cunningham, M.~T. Goodrich, S.~G. Kobourov, and
  L.~Trott.
\newblock {Force-directed Lombardi-style graph drawing}.
\newblock {\em Proc. 19th Internat. Symp. Graph Drawing (GD 2011)},
  pp.~320{--}331. Springer, Lect. Notes Comput. Sci. 7034, 2011,
  \href{http://dx.doi.org/10.1007/978-3-642-25878-7_31}%
{doi:\nolinkurl{10.1007/978-3-642-25878-7_31}}.

\bibitem{DonFuXu-CVS-13}
W.~Dong, X.~Fu, G.~Xu, and Y.~Huang.
\newblock {An improved force-directed graph layout algorithm based on aesthetic
  criteria}.
\newblock {\em Comput. Vis. Sci.} 16(3):139{--}149, 2013,
  \href{http://dx.doi.org/10.1007/s00791-014-0228-5}%
{doi:\nolinkurl{10.1007/s00791-014-0228-5}}.

\bibitem{Dud-PhD-16}
P.~M. Dudas.
\newblock {\em {The Impact of Visual Aesthetics on the Utility, Affordance, and
  Readability of Network Graphs}}.
\newblock Ph.D. thesis, University of Pittsburgh, 2016,
  \url{https://d-scholarship.pitt.edu/26607/}.

\bibitem{DunEppGoo-JoCG-18}
C.~A. Duncan, D.~Eppstein, M.~T. Goodrich, S.~G. Kobourov, M.~L{\"o}ffler, and
  M.~N{\"o}llenburg.
\newblock {Planar and poly-arc Lombardi drawings}.
\newblock {\em J. Comput. Geom.} 9(1):328{--}355, 2018,
  \href{http://dx.doi.org/10.20382/jocg.v9i1a11}%
{doi:\nolinkurl{10.20382/jocg.v9i1a11}},
  \href{https://www.ams.org/mathscinet-getitem?mr=3855883}%
{MR3855883}.

\bibitem{DunEppGoo-JGAA-12}
C.~A. Duncan, D.~Eppstein, M.~T. Goodrich, S.~G. Kobourov, and
  M.~N{\"o}llenburg.
\newblock {Lombardi drawings of graphs}.
\newblock {\em J. Graph Algorithms Appl.} 16(1):85{--}108, 2012,
  \href{http://dx.doi.org/10.7155/jgaa.00251}%
{doi:\nolinkurl{10.7155/jgaa.00251}},
  \href{https://www.ams.org/mathscinet-getitem?mr=2872431}%
{MR2872431}.

\bibitem{DunEppGoo-DCG-13}
C.~A. Duncan, D.~Eppstein, M.~T. Goodrich, S.~G. Kobourov, and
  M.~N{\"o}llenburg.
\newblock {Drawing trees with perfect angular resolution and polynomial area}.
\newblock {\em Discrete Comput. Geom.} 49(2):157{--}182, 2013,
  \href{http://dx.doi.org/10.1007/s00454-012-9472-y}%
{doi:\nolinkurl{10.1007/s00454-012-9472-y}},
  \href{https://www.ams.org/mathscinet-getitem?mr=3017904}%
{MR3017904}.

\bibitem{EfrErtKob-JGAA-07}
A.~Efrat, C.~Erten, and S.~G. Kobourov.
\newblock {Fixed-location circular arc drawing of planar graphs}.
\newblock {\em J. Graph Algorithms Appl.} 11(1):145{--}164, 2007,
  \href{http://dx.doi.org/10.7155/jgaa.00140}%
{doi:\nolinkurl{10.7155/jgaa.00140}},
  \href{https://www.ams.org/mathscinet-getitem?mr=2354167}%
{MR2354167}.

\bibitem{Epp-DCG-14}
D.~Eppstein.
\newblock {A M{\"o}bius-invariant power diagram and its applications to soap
  bubbles and planar Lombardi drawing}.
\newblock {\em Discrete Comput. Geom.} 52(3):515{--}550, 2014,
  \href{http://dx.doi.org/10.1007/s00454-014-9627-0}%
{doi:\nolinkurl{10.1007/s00454-014-9627-0}},
  \href{https://www.ams.org/mathscinet-getitem?mr=3257673}%
{MR3257673}.

\bibitem{Epp-CCCG-19}
D.~Eppstein.
\newblock {Bipartite and series-parallel graphs without planar Lombardi
  drawings}.
\newblock {\em Proc. 31st Canad. Conf. Comput. Geom. (CCCG 2019)},
  pp.~17{--}22, 2019.

\bibitem{Fej-BAMS-48}
L.~Fejes~T{\'o}th.
\newblock {Approximation by polygons and polyhedra}.
\newblock {\em Bull. Amer. Math. Soc.} 54:431{--}438, 1948,
  \href{http://dx.doi.org/10.1090/S0002-9904-1948-09022-X}%
{doi:\nolinkurl{10.1090/S0002-9904-1948-09022-X}},
  \href{https://www.ams.org/mathscinet-getitem?mr=24640}%
{MR24640}.

\bibitem{HarUhl-PNAS-53}
F.~Harary and G.~E. Uhlenbeck.
\newblock {On the number of Husimi trees, I}.
\newblock {\em Proc. Natl. Acad. Sci. U.S.A.} 39(4):315{--}322, 1953,
  \href{http://dx.doi.org/10.1073/pnas.39.4.315}%
{doi:\nolinkurl{10.1073/pnas.39.4.315}},
  \href{https://www.ams.org/mathscinet-getitem?mr=0053893}%
{MR0053893}.

\bibitem{HobRic-03}
R.~C. Hobbs and J.~Richards.
\newblock {\em {Mark Lombardi: Global Networks}}.
\newblock Independent Curators International, 2003.

\bibitem{Hus-JCP-50}
K.~Husimi.
\newblock {Note on Mayers' theory of cluster integrals}.
\newblock {\em J. Chem. Phys.} 18(5):682{--}684, 1950,
  \href{http://dx.doi.org/10.1063/1.1747725}%
{doi:\nolinkurl{10.1063/1.1747725}},
  \href{https://www.ams.org/mathscinet-getitem?mr=0038903}%
{MR0038903}.

\bibitem{Kel-BAMS-26}
O.~D. Kellogg.
\newblock {The second edition of the Hurwitz{--}Courant Funktionentheorie}.
\newblock {\em Bull. Amer. Math. Soc.} 32(3):288{--}292, 1926,
  \href{http://dx.doi.org/10.1090/S0002-9904-1926-04215-4}%
{doi:\nolinkurl{10.1090/S0002-9904-1926-04215-4}},
  \href{https://www.ams.org/mathscinet-getitem?mr=1561208}%
{MR1561208}.

\bibitem{KinKobLof-JoCG-19}
P.~Kindermann, S.~G. Kobourov, M.~L{\"o}ffler, M.~N{\"o}llenburg, A.~Schulz,
  and B.~Vogtenhuber.
\newblock {Lombardi drawings of knots and links}.
\newblock {\em J. Comput. Geom.} 10(1):444{--}476, 2019,
  \href{http://dx.doi.org/10.20382/jocg.v10i1a15}%
{doi:\nolinkurl{10.20382/jocg.v10i1a15}},
  \href{https://www.ams.org/mathscinet-getitem?mr=4039890}%
{MR4039890}.

\bibitem{LeiMoi-DCG-10}
T.~Leighton and A.~Moitra.
\newblock {Some results on greedy embeddings in metric spaces}.
\newblock {\em Discrete Comput. Geom.} 44(3):686{--}705, 2010,
  \href{http://dx.doi.org/10.1007/s00454-009-9227-6}%
{doi:\nolinkurl{10.1007/s00454-009-9227-6}},
  \href{https://www.ams.org/mathscinet-getitem?mr=2679063}%
{MR2679063}.

\bibitem{MeeWal-JCAM-95}
D.~S. Meek and D.~J. Walton.
\newblock {Approximating smooth planar curves by arc splines}.
\newblock {\em J. Comput. Appl. Math.} 59(2):221{--}231, 1995,
  \href{http://dx.doi.org/10.1016/0377-0427(94)00029-Z}%
{doi:\nolinkurl{10.1016/0377-0427(94)00029-Z}},
  \href{https://www.ams.org/mathscinet-getitem?mr=1346015}%
{MR1346015}.

\bibitem{PlaTseShy-ICTME-20}
S.~Plankovskyy, Y.~Tsegelnyk, O.~Shypul, A.~Pankratov, and T.~Romanova.
\newblock {Cutting irregular objects from the rectangular metal sheet}.
\newblock {\em Integrated Computer Technologies in Mechanical Engineering},
  pp.~150{--}157. Springer, Adv. Intell. Syst. Comput. 1113, 2020,
  \href{http://dx.doi.org/10.1007/978-3-030-37618-5_14}%
{doi:\nolinkurl{10.1007/978-3-030-37618-5_14}}.

\bibitem{PurHamNol-GD-12}
H.~C. Purchase, J.~Hamer, M.~N{\"o}llenburg, and S.~G. Kobourov.
\newblock {On the Usability of Lombardi Graph Drawings}.
\newblock {\em Proc. 20th Internat. Symp. Graph Drawing (GD 2012)},
  pp.~451{--}462. Springer, Lect. Notes Comput. Sci. 7704, 2012,
  \href{http://dx.doi.org/10.1007/978-3-642-36763-2_40}%
{doi:\nolinkurl{10.1007/978-3-642-36763-2_40}}.

\bibitem{Schu-JGAA-15}
A.~Schulz.
\newblock {Drawing graphs with few arcs}.
\newblock {\em J. Graph Algorithms Appl.} 19(1):393{--}412, 2015,
  \href{http://dx.doi.org/10.7155/jgaa.00366}%
{doi:\nolinkurl{10.7155/jgaa.00366}},
  \href{https://www.ams.org/mathscinet-getitem?mr=3395765}%
{MR3395765}.

\bibitem{Sch-79}
H.~Schwerdtfeger.
\newblock {\em {Geometry of Complex Numbers}}.
\newblock Dover, 1979.

\bibitem{WanLinFan-CAD-17}
Z.-J. Wang, X.~Lin, M.-E. Fang, B.~Yao, Y.~Peng, H.~Guan, and M.~Guo.
\newblock {$\textsc{Re2l}$: an efficient output-sensitive algorithm for
  computing Boolean operations on circular-arc polygons and its applications}.
\newblock {\em Comput.-Aided Des.} 83:1{--}14, 2017,
  \href{http://dx.doi.org/10.1016/j.cad.2016.07.004}%
{doi:\nolinkurl{10.1016/j.cad.2016.07.004}},
  \href{https://www.ams.org/mathscinet-getitem?mr=3577906}%
{MR3577906}.

\bibitem{WeiJutAur-EuroCG-18}
B.~Wei{\ss}, B.~J{\"u}ttler, and F.~Aurenhammer.
\newblock {Mitered offsets and straight skeletons for circular arc polygons}.
\newblock {\em Eur. Worksh. Comput. Geom. (EuroCG 2018)}, pp.~52:1{--}52:6,
  2018,
  \url{https://conference.imp.fu-berlin.de/eurocg18/download/paper_52.pdf}.

\bibitem{XuRooPas-TVCG-12}
K.~Xu, C.~Rooney, P.~J. Passmore, D.~Ham, and P.~H. Nguyen.
\newblock {A user study on curved edges in graph visualization}.
\newblock {\em IEEE Trans. Vis. Comput. Graph.} 18(12):2449{--}2456, 2012,
  \href{http://dx.doi.org/10.1109/TVCG.2012.189}%
{doi:\nolinkurl{10.1109/TVCG.2012.189}}.

\end{thebibliography}

\end{document}